\documentclass[11pt]{article}
\usepackage{amsmath, mathtools}
\usepackage{amsthm}
\usepackage{amssymb}
\usepackage{graphicx}
\usepackage{appendix}
\usepackage{color}

\usepackage{hyperref}

\usepackage{float}

\usepackage{multirow}

\usepackage{booktabs}

\newtheorem{theorem}{Theorem}[section]

\newtheorem{corollary}[theorem]{Corollary}

\numberwithin{equation}{section}
\begin{document}

\title{Hedging and machine learning driven crude oil data analysis using a refined Barndorff-Nielsen and Shephard model}
\author{ Humayra Shoshi \footnote{Email: humayra.shoshi@ndsu.edu}, Indranil SenGupta\footnote{Email: indranil.sengupta@ndsu.edu} \\ Department of Mathematics \\ North Dakota State University \\ Fargo, North Dakota, USA.}
\date{\today}
\maketitle
\begin{abstract}

In this paper, a refined Barndorff-Nielsen and Shephard (BN-S) model is implemented to find an optimal hedging strategy for commodity markets. The refinement of the BN-S model is obtained with various machine and deep learning algorithms. The refinement leads to the extraction of a deterministic parameter from the empirical data set. The problem is transformed to an appropriate classification problem with a couple of different approaches- the volatility approach and the duration approach. The analysis is implemented to the Bakken crude oil data and the aforementioned deterministic parameter is obtained for a wide range of data sets. With the implementation of this parameter in the refined model, the resulting model performs much better than the classical BN-S model.

\end{abstract}
\section{Introduction}
\label{sec1}

Price risk in commodity trading refers to fluctuation in the price of asset. To reduce price risk, traders hedge the commodity price with commodity derivatives such as futures, options, or swaps. Hedging is an act of taking opposite position in the similar market to reduce the price risk. With appropriate hedging of the underlying position, the loss from one market is offset by another market.

A commodity of fundamental importance is crude oil. Consequently a study of the fluctuation of crude oil price time series is of utmost importance (see \cite{Grinblatt, He}). This allows to evaluate the potential impacts of its shocks in several economies and on other financial assets. In \cite{Tabak}, the authors analyze the efficiency of crude oil markets by means of estimating the fractal structure of these time series. In \cite{Sensoy}, it is shown that the efficiency of energy futures markets is time-varying and changes drastically over the sample period. In particular, for futures contracts with one to four months to maturities, crude oil and gasoline are found to be more efficient compared to others. In  \cite{Sircar1}, the authors discuss how the traditional oil producers may react in counter-intuitive ways in face of competition from alternative energy sources. The paper considers the big decline in oil prices, from around \$110 per barrel in June 2014 to less than \$40 in March 2016, and shows the significance of competition between different energy sources. With the ongoing COVID-19 pandemic situation, this analysis is very relevant. 
In \cite{Roberts}, the authors present a sequential hypothesis testing on two streams of observations that are driven by L\'evy processes. After that, machine learning algorithms are implemented to analyze the oil price dynamics for the Bakken region in the United States.

A frequently used stochastic volatility model for the commodity market analysis is the Barndorff-Nielsen and Shephard (BN-S) model (see \cite{BN1, BN-S1, BN-S2, BJS, Semere, Issaka}). This model finds various applications in the derivative and commodity market. In recent literature, the BN-S model is implemented to find an optimal hedging strategy for the oil commodity (see \cite{SWW, Wil2}).  In \cite{Tat}, the BN-S model is implemented to the analysis of the S\&P 500 market using a $K$-component mixture of regressions model. In spite of having a lot of advantages, the classical BN-S model has some major disadvantages including a ``short-range dependence". In the paper \cite{recent}, a refinement of the BN-S model is proposed. It is shown that a machine learning driven refined BN-S model can be used as an improvement of the classical BN-S model.  In the recent paper \cite{Hui}, it is shown that the refined BN-S model in a fuzzy environment solves the long-term dependence problem of the classical model, and thus effectively analyzes the random dynamic characteristics of stock index option price time series.
All these analyses are further improved in \cite{Roberts1}, where a machine learning driven sequential hypothesis testing is implemented to refine the BN-S model. In both the papers (\cite{Roberts1} and \cite{recent}), machine learning based techniques are implemented for extracting a \emph{deterministic component}  from the commodity price processes. Also, the refined BN-S model is shown to incorporate \emph{long range dependence} without actually changing the model.

In this paper, we investigate the refinement of the BN-S model by analyzing the underlying data set with a couple of different approaches- (1) volatility approach, and (2) duration approach. In effect, these approaches provide a ``jump-detection technique" for a financial time series. The papers \cite{D1, D2, V1, V2}, discuss various motivations for these approaches. In \cite{D1}, it is observed that by fitting the log-periodic power law equation to a financial time series, it is possible to predict the event of a crash. The paper investigates the financial crisis of 2008, with the log-periodic power law. In \cite{D2}, drawdowns, defined as the loss from the last local maximum to the next local minimum, is introduced. It is shown that drawdowns can be used as a natural measure of real market risks than the variance, the value-at-risk or other measures based on fixed time scale distributions of returns. It is shown that very large drawdowns belong to a different class of their own and call for a specific amplification mechanism. In \cite{V1}, drawdowns are implemented and crashes are classified as either events of an endogenous origin preceded by speculative bubbles or as events of exogenous origins associated to external shocks. However, the proposed classification does not rule out the existence of other precursory signals in the absence of so-called log-periodic power law signatures. In \cite{V2}, the price volatility before, during, and after financial asset bubbles are investigated for possible commonalities. It is also empirically investigated whether volatility may be used as an indicator or an early warning signal of an unsustainable price increase and the associated crash.

The organization of the paper is as follows. In Section \ref{sec2}, a refined BN-S model is presented. Some useful properties of variance swaps  with respect to the refined BN-S model is studied. In addition, a quadratic hedging procedure is discussed. In Section \ref{sec3}, the data set is provided, and then two procedures, the volatility approach and  the duration approach, in the classification problem are introduced. Various numerical results are also provided in that section. Finally, a brief conclusion is provided in Section \ref{sec4}.

\section{Refined Barndorff-Nielsen and Shephard model and related results}
\label{sec2}

Many models in recent literature try to capture the stochastic behavior of time series. For example, in the case of the Barndorff-Nielsen and Shephard (BN-S) model, the stock or commodity price $S= (S_t)_{t \geq 0}$ on some risk-neutral filtered probability space is modeled by

\begin{equation}
\label{1}
S_t= S_0 \exp (X_t),
\end{equation}
\begin{equation}
\label{2}
dX_t = b_t \,dt + \sigma_t\, dW_t + \rho \,dZ_{\lambda t},  \quad \text{with} \quad b_t=  (r- \lambda \kappa (\rho)- \frac{1}{2} \sigma_t^2),
\end{equation}
\begin{equation}
\label{3}
d\sigma_t ^2 = -\lambda \sigma_t^2 \,dt + dZ_{\lambda t}, \quad \sigma_0^2 >0,
\end{equation}
where the parameters $ \rho, \lambda \in \mathbb{R}$ with $\lambda >0$, and $\rho \leq 0$. Here $r$ is the risk-free interest rate where a stock or commodity is traded up to a fixed horizon date $T$. In the expression for $b_t$, the cumulant transform for $Z_1$ under the new measure is denoted as $\kappa(\cdot)$. In this model $W_t$ is a Brownian motion and the process $Z_{t}$ is a subordinator. 
For a refined BN-S model (see \cite{recent}) the stock or commodity price $S= (S_t)_{t \geq 0}$ on some risk-neutral filtered probability space $(\Omega, \mathcal{F}, (\mathcal{F}_t)_{0 \leq t \leq T}, \mathbb{Q})$ is modeled by \eqref{1}, with

\begin{equation}
\label{2new}
dX_t = b_t \,dt + \sigma_t\, dW_t + \rho\left( (1-\theta) \,dZ_{\lambda t}+ \theta dZ^{(b)}_{\lambda t}\right), 
\end{equation}
where $\theta \in [0,1]$ is a \emph{deterministic} parameter, and $b_t$ is given by \eqref{2}. Machine learning algorithms are implemented to determine the value of $\theta$. The process $Z^{(b)}$ in \eqref{2new} is a subordinator that is independent of $Z$. In addition, $Z^{(b)}$ has greater intensity than the subordinator $Z$. $W$, $Z$ and $Z^{(b)}$ are assumed to be independent, and $(\mathcal{F}_t)$ is assumed to be the usual augmentation of the filtration generated by $(W, Z, Z^{(b)} )$. 

When empirical financial data is considered for a long duration of time, it is observed that the ``big" fluctuations can only be modeled poorly with the incorporation of a single jump term. This is one of the known disadvantages of the classical BN-S model. Consequently, in order to model those ``big" fluctuations for a longer period of time, it is natural to use two or more L\'evy processes. Consequently, as a simple case, we incorporate two jump terms- $Z$ and $Z^{(b)}$- for the revised BN-S model.

In this case \eqref{3} is given by
\begin{equation}
\label{4new}
d\sigma_t ^2 = -\lambda \sigma_t^2 \,dt + (1- \theta') dZ_{\lambda t} + \theta' dZ^{(b)}_{\lambda t} , \quad \sigma_0^2 >0,
\end{equation}
where, as before, $\theta' \in [0,1]$ is \emph{deterministic}. For simplicity, we assume $\theta= \theta'$ for the rest of this paper. 

As shown in \cite{recent}, the dynamics given by \eqref{1}, \eqref{2new}, and \eqref{4new} incorporates a long-range dependence. If the jump measures associated with the subordinators $Z$ and $Z^{(b)}$ are $J_Z$ and $J^{(b)}_Z$ respectively, and $J(s)= \int_0^s \int_{\mathbb{R}^+} J_Z(\lambda d\tau, dy)$, $J^{(b)}(s)= \int_0^s \int_{\mathbb{R}^+} J^{(b)}_Z(\lambda d\tau, dy)$; then for the log-return of the improved BN-S model given by \eqref{1}, \eqref{2new}, and \eqref{4new},
\begin{align}
\label{corrBNSimproved}
\text{Corr}(X_t, X_s)= \frac{\int_0^s \sigma_{\tau}^2 d\tau + \rho^2 (1-\theta)^2 J(s) + \rho^2 \theta^2 J^{(b)}(s)}{ \sqrt{\alpha(t) \alpha(s)}},
\end{align}
for $t>s$, where
$\alpha(\nu) = \int_0^{\nu} \sigma_{\tau}^2 d\tau + \nu\rho^2 \lambda ((1-\theta)^2 \text{Var}(Z_1)+  \theta^2 \text{Var}(Z^{(b)}_1)) $.

We observe that the solution of \eqref{4new} can be written as
\begin{equation}
\label{se43}
\sigma_t^2= e^{-\lambda t}\sigma_0^2 + (1-\theta)\int_0^t e^{-\lambda (t-s)}\, dZ_{\lambda s}+ \theta \int_0^t e^{-\lambda (t-s)}\, dZ^{(b)}_{\lambda s}.
\end{equation}
This enforces \emph{positivity} of $\sigma_t^2$. Thus, the process $\sigma_t^2$ is strictly positive and it is bounded from below by the deterministic function $e^{-\lambda t}\sigma_0^2$. The instantaneous variance of log returns is given by $$(\sigma_t^2+  \rho^2 (1-\theta)^2\lambda \text{Var}[Z_1]+ \rho^2 \theta^2\lambda \text{Var}[Z^{(b)}_1])\,dt,$$ and therefore simple calculation shows that the continuous realized variance in the interval $[0,T]$ is
\begin{equation*}
\sigma_R^2=\frac{1}{T}\int_0^T \sigma_t^2\,dt + \rho^2 (1-\theta)^2 \lambda \text{Var}[Z_1]+ \rho^2 \theta^2 \lambda \text{Var}[Z^{(b)}_1].
\end{equation*}

For the rest of this section we develop a procedure to show an effective hedging algorithm using the refined BN-S model. In Subsection \ref{sub21}, we briefly introduce some results related to the variance swap. In Subsection \ref{sub22} we develop results related to hedging algorithm where variance swaps and some \emph{specific} options are used.

\subsection{Variance swap pricing with the refined BN-S model}
\label{sub21}
A variance swap is  a forward contract on realized variance (see \cite{Semere, Semere2, Issaka}). The payoff of variance swap at the maturity $T$ is given by  $N(\sigma_R^2-\text{K}_{\text{Var}})$,
where $K_\text{Var}$ is the annualized delivery price or exercise price of the variance swap, and $N$ is the notional amount of the dollars per annualized volatility point squared. 
Without loss of generality we take $N=1$. The arbitrage free price of the variance swap is the expectation of the present value of the payoff in the risk-neutral world and it is given by $E^{\mathbb{Q}}\left[e^{-r(T-t)}(\sigma_R^2-K_\text{Var})| \mathcal{F}_t \right]$,  $0 \leq t \leq T$, where $\mathcal{F}_t$ is the $\sigma$-field generated by the history of the process up to time $t$. When $\mathcal{F}_t$ is given and $s \geq t$ a similar derivation as in \eqref{se43} gives
\begin{equation}
\label{se4300}
\sigma_s^2= e^{-\lambda (s-t)}\sigma_t^2 + (1-\theta) \int_t^s e^{-\lambda (s-u)}\, dZ_{\lambda u}+ \theta \int_t^s e^{-\lambda (s-u)}\, dZ^{(b)}_{\lambda u}.
\end{equation}

We denote $V_t= \int_0^t\sigma_u^2\,du$. For a fixed horizon date $T$, we consider $P_\text{Var}(t, \sigma_t^2, V_t)$ as a function of $t$, $\sigma_t^2$ and  $V_t$ with the final condition (independent of $S$) given by
\begin{equation*}
P_\text{Var}(T, \sigma_T^2, V_T)=  \sigma_R^2-\text{K}_{\text{Var}}= \frac{V_T}{T}- \text{K}_{\text{Var}}.
\end{equation*}
Using \eqref{se4300} we obtain 
\begin{align} 
\label{se4400}
\sigma_R^2 & = \frac{1}{T}\int_0^T \sigma_s^2\,ds + \rho^2 (1-\theta)^2 \lambda \text{Var}[Z_1] + \rho^2 \theta^2 \lambda \text{Var}[Z^{(b)}_1] \nonumber \\
&= \frac{1}{T}\left(\int_0^t \sigma_s^2\,ds + \int_t^T \sigma_s^2\,ds \right)+\rho^2 (1-\theta)^2 \lambda \text{Var}[Z_1] + \rho^2 \theta^2 \lambda \text{Var}[Z^{(b)}_1] \nonumber \\
&= \frac{1}{T} \left(V_t + \frac{1}{\lambda}(1-e^{-\lambda (T-t)})\sigma_t^2+\frac{1-\theta}{\lambda}\int_t^T\left(1-e^{-\lambda(T-s)}\right)dZ_{\lambda s} + \frac{\theta}{\lambda}\int_t^T\left(1-e^{-\lambda(T-s)}\right)dZ^{(b)}_{\lambda s} \right) \nonumber \\
& + \rho^2 (1-\theta)^2 \lambda \text{Var}[Z_1] + \rho^2 \theta^2 \lambda \text{Var}[Z^{(b)}_1].
\end{align}

Based on this result we can prove the following theorem.
\begin{theorem}
\label{arbitr1}
The arbitrage free price of the variance swap, with respect to the risk neutral measure $\mathbb{Q}$, is given by
\begin{align*}
& \text{P}_{\text{Var}}(t, \sigma_t^2, V_t) = e^{-r(T-t)} \left[V_t + (T-t) \left(\kappa_1(1-\theta)+ \kappa_1^{(b)}\theta\right) \right.\\
& \left.+ \frac{1}{\lambda}\left(1-e^{-\lambda (T-t)}\right)\left(\sigma_t^2- \kappa_1(1-\theta)- \kappa_1^{(b)}\theta\right) 
+ \rho^2 (1-\theta)^2 \lambda \kappa_2 + \rho^2 \theta^2 \lambda \kappa_2^{(b)}-\text{K}_{\text{Var}}\right],
\end{align*}
where $\kappa_1$ and $\kappa_2$ are the first cumulant (i.e., the expected value) and  the second cumulant (i.e., the variance) of $Z_1$ respectively; and $\kappa^{(b)}_1$ and $\kappa^{(b)}_2$ are the first cumulant (i.e., the expected value) and  the second cumulant (i.e., the variance) of $Z^{(b)}_1$ respectively.
%that can be obtained by taking first derivative of the cumulant generating function.  
\end{theorem}
\begin{proof}
The conditional expected value, given $\mathcal{F}_t$, of  equation \eqref{se4400} gives the value
\begin{align}
\label{se49}
E(\sigma_R^2| \mathcal{F}_t)&= \frac{1}{T} (V_t+ \frac{1}{\lambda}(1-e^{-\lambda (T-t)})\sigma_t^2+\frac{(1-\theta)\kappa_1}{\lambda}  \int_t^T\left(1-e^{-\lambda(T-s)}\right)\lambda\, ds \nonumber \\
&+ \frac{\theta \kappa^{(b)}_1}{\lambda}  \int_t^T\left(1-e^{-\lambda(T-s)}\right)\lambda\, ds) + \rho^2 (1-\theta)^2 \lambda \text{Var}[Z_1] + \rho^2 \theta^2 \lambda \text{Var}[Z^{(b)}_1] \nonumber \\
&=\frac{1}{T}(V_t+\frac{1}{\lambda}(1-e^{-\lambda (T-t)})\sigma_t^2+\kappa_1 (1-\theta)\left(T-t-\frac{1}{\lambda}\left(1-e^{-\lambda (T-t)}\right)\right) \nonumber \\
& +\kappa_1^{(b)}\theta\left(T-t-\frac{1}{\lambda}\left(1-e^{-\lambda (T-t)}\right)\right) )+ \rho^2 (1-\theta)^2 \lambda \kappa_2 + \rho^2 \theta^2 \lambda \kappa_2^{(b)}.
\end{align}
Hence the theorem follows from simplification of \eqref{se49}.
\end{proof}

\subsection{Quadratic optimal hedging strategy}
\label{sub222222}

In this subsection we provide a brief formal introduction to optimal hedging strategy in terms of quadratic hedging. Quadratic hedging is a hedging strategy which minimizes the hedging error in the mean square sense. Some basic cases for the quadratic hedging are provided in \cite{cont, Wil2}. Consider a risk-neutral measure $\mathbb{Q}$, and assume that $(S_t)_{t \in [0,T]}$ given by $S_t= \exp(rt+ X_t)$, where $X_t$ is a L\'evy process on $(\Omega, \mathcal{F}, \mathcal{F}_t, \mathbb{Q})$. Let $W_t$ be a Brownian motion with respect to $\mathbb{Q}$. In the following we denote the discounted functions with ``hats". We also denote the risk-free interest rate by $r$. For instance, discounted commodity price will be given by $\hat{S}_t= e^{-rt}S_t$. 

The process $\hat{S}_t$ can be written as the stochastic exponential of another L\'evy process $Z_t$ as $d \hat{S}_t= \hat{S}_{t-}\,dZ_t$, where $Z$ is a martingale with jumps greater than $-1$. Consider a self-financing strategy $(\phi_t^0, \phi_t)_{t \in [0,T]}$. The terminal payoff of such strategy is given by 
\begin{equation*}
G_T(\phi)= \int_0^T r \phi_t^0 \,dt + \int_0^T \phi_t S_{t-} \,dZ_t. 
\end{equation*}
If the jump measure associate with $X$ is given by $J_X(\cdot, \cdot)$, then as obtained in \cite{Wil2},
\begin{align*}
\hat{G}_T(\phi) =  \int_0^T \phi_t S_{t-} \sigma dW_t +\int_0^T \int_{\mathbb{R}} \tilde{J}_Z (dt, dx) (e^z-1)\phi_t S_{t-}. 
\end{align*}
Denote $\hat{\mathcal{S}}= \{ \phi \text{  predictable  and  } E| \int_0^T \phi_t\, d\hat{S}_t|^2 < \infty \}.$ Given the initial capital $\Pi_0$, and a random variable $\mathcal{H}$, the quadratic hedging problem is given by (see \cite{cont})
\begin{equation}
\label{qhedge}
\inf_{\phi \in L^2(\hat{S})} E^{\mathbb{Q}} |\hat{G}_T(\phi)+ \Pi_0- \hat{\mathcal{H}}|^2. 
\end{equation}
From this, it follows that the hedging error is $\hat{G}_T(\phi)+ \Pi_0- \hat{\mathcal{H}}$. 
From the construction of $\hat{G}_T(\phi)$, we obtain $E^{\mathbb{Q}}(\hat{G}_T(\phi))=0$. Consequently, the expectation of hedging error is $\Pi_0- E^{\mathbb{Q}}  [\hat{\mathcal{H}}]$. Thus, the optimal value for the initial capital is $\Pi_0= E^{\mathbb{Q}}  [\hat{\mathcal{H}}]$.

\subsection{Quadratic hedging under the refined BN-S model}
\label{sub22}

In this subsection, we show that there is an effective hedging procedure in relation to the refined BN-S model given by \eqref{1}, \eqref{2new}, and \eqref{4new}. With respect to $\mathbb{Q}$, the dynamics of $S_t$ is given by
\begin{align}
\label{stockdyn}
\frac{dS_t}{S_t} = rdt + \sigma_t\,dW_t + \int_{\mathbb{R}_{+}}(e^{\rho (1-\theta) x}-1)\tilde{J}_{Z}(\lambda dt, dx)+ \int_{\mathbb{R}_{+}}(e^{\rho \theta x}-1)\tilde{J}_{Z^{(b)}}(\lambda dt, dx),
\end{align}
where we assume that random measures associated with the jumps of $Z$ and $Z^{(b)}$, and L\'evy densities of $Z$ and $Z^{(b)}$ are given by $J_Z$, $J_{Z^{(b)}}$, and $\nu_Z$, $\nu_{Z^{(b)}}$, respectively. The compensator for $J_Z (\lambda dt, dx)$ is given by $\lambda \nu_Z(dx) \,dt$ and we define $\tilde{J}_Z (\lambda dt, dx)= J_Z (\lambda dt, dx)-\lambda \nu_Z(dx) \,dt$. Similarly, the compensator for $J_{Z^{(b)}}(\lambda dt, dx)$ is given by $\lambda \nu_Z^{(b)}(dx) \,dt$ and we define $\tilde{J}_{Z^{(b)}} (\lambda dt, dx)= J_{Z^{(b)}} (\lambda dt, dx)-\lambda \nu_{Z^{(b)}}(dx) \,dt$. 

As introduced in \cite{SWW} and \cite{Wil2}, we consider a ``stable" commodity $Y_t$ given by (with respect to $\mathbb{Q}$) a geometric Brownian motion
\begin{equation}
\label{Yt}
dY_t= Y_t(r\,dt+ \sigma \,d\tilde{W}_t),
\end{equation}
with $d\tilde{W}_t \cdot dW_t= \rho'\,dt$, with $W_t$ defined in \eqref{2new} (same as in \eqref{stockdyn}), and $\sigma>0$ a constant.

\begin{theorem}
\label{biggy1}
Consider a European option with payoff $H(Y_T)$ where $H: \mathbb{R}_{+} \to \mathbb{R}$. 
%verifies 
%\begin{equation}
%\label{lipsc}
%|H(x)-H(y)| \leq K|x-y|,
%\end{equation}
%for some $K>0$. 
Then the risk-minimizing quadratic hedge amounts to holding a position of the underlying $S$ equal to $\phi_t= \Delta(t, S_t, Y_t)$, where
\begin{align}
\label{optphit}
\Delta(t, S_t, Y_t)= \frac{\rho' \sigma \sigma_t \frac{Y_t}{S_t} \frac{\partial C}{\partial Y}+A+B}{\sigma_t^2 +\lambda \int_{\mathbb{R}_{+}} (e^{\rho (1-\theta) x}-1)^2  \nu_Z(dx) + \lambda \int_{\mathbb{R}_{+}} (e^{\rho \theta x}-1)^2  \nu_Z^{(b)}(dx) },
\end{align}
where $C$ is the Black-Scholes price of the option written on $Y$, and
\begin{equation}
\label{a}
A= \frac{\lambda (1-\theta)}{S_t}\int_{\mathbb{R}_{+}} \left(P(t, \sigma_{t}^2+x, V_t)-P(t, \sigma_{t}^2, V_{t})\right)(e^{\rho (1-\theta) x}-1) \nu_Z(dx),
\end{equation}
\begin{equation}
\label{a1}
B= \frac{\lambda \theta}{S_t}\int_{\mathbb{R}_{+}} \left(P(t, \sigma_{t}^2+x, V_t)-P(t, \sigma_{t}^2, V_{t})\right)(e^{\rho \theta x}-1) \nu_Z^{(b)}(dx).
\end{equation}

\end{theorem}

\begin{proof}
From \eqref{stockdyn}, it is clear that the discounted commodity price $\hat{S}_t= e^{-rt} S_t$ is a martingale with respect to $\mathbb{Q}$. We consider a self financing strategy $(\phi_t^0, \phi_t)$ with $\phi \in L^2(\hat{S})$. The discounted value of the portfolio ($\hat{\Pi}$) is then a martingale with terminal value given by
\begin{align}
\label{alig1}
\hat{\Pi}_T(\phi) & = \int_0^T \phi_t\, d\hat{S}_t \nonumber \\
& =  \int_0^T \phi_t \hat{S}_t \left(  \sigma_t\,dW_t + \int_{\mathbb{R}_{+}}(e^{\rho (1-\theta) x}-1)\tilde{J}_{Z}(\lambda dt, dx)+ \int_{\mathbb{R}_{+}}(e^{\rho \theta x}-1)\tilde{J}_{Z^{(b)}}(\lambda dt, dx)\right) \nonumber \\
& =  \int_0^T \phi_t \hat{S}_t \sigma_t\,dW_t +  \int_0^{T} \phi_t \hat{S}_t \left(\int_{\mathbb{R}_{+}}(e^{\rho (1-\theta) x}-1)\tilde{J}_{Z}(\lambda dt, dx)+ \int_{\mathbb{R}_{+}}(e^{\rho \theta x}-1)\tilde{J}_{Z^{(b)}}(\lambda dt, dx)\right).
\end{align}
The arbitrage-free price of the option written on the commodity $Y$ with payoff $H(Y_T)$ is given by
\begin{equation*}
C(t,Y)= e^{-r(T-t)}E^{\mathbb{Q}}[H(Y_T)|Y_t=Y].
\end{equation*}
We denote $\hat{C}(t,Y)= e^{-rt}C(t,Y)$ and $\Pi_{01}= \hat{C}(0,Y_0)= e^{-r T} E^{\mathbb{Q}}[H(Y_T)]$. Then, by It\^o formula we obtain 
\begin{align}
\label{alig2}
\hat{C}(t,Y_t)-\Pi_{01}= \int_0^t \frac{\partial C}{\partial Y}(u,Y_u) \hat{Y}_u\sigma \,d\tilde{W}_u.
\end{align}
On the other hand, if we consider a variance swap written on $S_t$, and denote $\hat{P}(t, \sigma_t^2,V_t)= e^{-rt}P(t, \sigma_t^2,V_t)$, $\hat{\tilde{P}}(t, \sigma_t^2,V_t)= e^{-rt}\tilde{P}(t, \sigma_t^2,V_t)$, and $\Pi_{02}= e^{-rT} \tilde{P}(0,\sigma_0^2, V_0)=P(0,\sigma_0^2, V_0)$, then, using It\^o formula we obtain:
\begin{align}
\label{alig3}
e^{-rT}\tilde{P}(t,\sigma_t^2, V_t)-\Pi_{02} & = (1-\theta) \int_0^{t} \int_{\mathbb{R}_{+}}\left(\hat{P}(s, \sigma_{s-}^2+x, V_s)-\hat{P}(s, \sigma_{s-}^2, V_{s})\right)  \tilde{J}_{Z}(\lambda ds, dx) \nonumber \\
& + \theta \int_0^{t} \int_{\mathbb{R}_{+}}\left(\hat{P}(s, \sigma_{s-}^2+x, V_s)-\hat{P}(s, \sigma_{s-}^2, V_{s})\right)  \tilde{J}_{Z^{(b)}}(\lambda ds, dx).
\end{align}
We denote $\Pi_0= \Pi_{01}+\Pi_{02}$, and $\epsilon(\phi, \Pi_0)= \hat{\Pi}_T(\phi)+ \Pi_{0}- \hat{C}(T,Y_T)- \hat{\tilde{P}}(T, \sigma_T^2 ,V_T)$. Note that $\tilde{P}(T,\sigma_T^2, V_T)= P(T,\sigma_T^2, V_T)$, and thus we have
\begin{equation}
\label{errorterm}
\epsilon(\phi, \Pi_0)= \hat{\Pi}_T(\phi)+ \Pi_{0}- \hat{C}(T,Y_T)- \hat{P}(T, \sigma_T^2 ,V_T).
\end{equation} Considering expressions in \eqref{alig2} and \eqref{alig3} at $t=T$, adding those, and subtracting from \eqref{alig1} we obtain
\begin{align*}
 \epsilon(\phi, \Pi_0) & =  \int_0^T \phi_t \hat{S}_t \sigma_t \,dW_t- \int_0^T \frac{\partial C}{\partial Y} \hat{Y}_t\sigma \,d\tilde{W}_t \\
& + \int_0^{T} \int_{\mathbb{R}_{+}}\left[\phi_t \hat{S}_t  (e^{\rho (1-\theta) x}-1)-(1-\theta)\left(\hat{P}(t, \sigma_{t}^2+x, V_t)-\hat{P}(t, \sigma_{t}^2, V_{t})\right) \right]\tilde{J}_{Z}(\lambda  dt, dx) \\
& + \int_0^{T} \int_{\mathbb{R}_{+}}\left[\phi_t \hat{S}_t  (e^{\rho \theta x}-1)-\theta\left(\hat{P}(t, \sigma_{t}^2+x, V_t)-\hat{P}(t, \sigma_{t}^2, V_{t})\right) \right]\tilde{J}_{Z^{(b)}}(\lambda  dt, dx).
\end{align*}

Using the isometry formula and observing $E^{\mathbb{Q}}[\epsilon(\phi, \Pi_0) ]=0$, we obtain the variance of $\epsilon(\phi, \Pi_0) $ as 

\begin{align*}
& E^{\mathbb{Q}}[\epsilon(\phi, \Pi_0)]^2  =  E^{\mathbb{Q}} \left[ \int_0^T \phi_t^2 \hat{S}_t^2  \sigma_t^2 \,dt\right]+ E^{\mathbb{Q}} \left[ \int_0^T \left(\frac{\partial C}{\partial Y}\right)^2 \hat{Y}_t^2 \sigma^2 \,dt \right]\\
& + E^{\mathbb{Q}} \left[\int_0^{T} \int_{\mathbb{R}_{+}}\left[\phi_t \hat{S}_t  (e^{\rho (1-\theta) x}-1)-(1-\theta) \left(\hat{P}(t, \sigma_{t}^2+x, V_t)-\hat{P}(t, \sigma_{t}^2, V_{t})\right)\right]^2 \lambda  \nu_Z(dx)\,dt \right] \\
&  + E^{\mathbb{Q}} \left[\int_0^{T} \int_{\mathbb{R}_{+}}\left[\phi_t \hat{S}_t  (e^{\rho \theta x}-1)- \theta \left(\hat{P}(t, \sigma_{t}^2+x, V_t)-\hat{P}(t, \sigma_{t}^2, V_{t})\right)\right]^2 \lambda  \nu_Z^{(b)}(dx)\,dt \right] \\
&- E^{\mathbb{Q}} \left[ 2 \rho' \sigma \int_0^T \phi_t \hat{S}_t \hat{Y}_t  \sigma_t \frac{\partial C}{\partial Y}  \,dt\right].
\end{align*}
The optimal (risk-minimizing) hedge is obtained by minimizing this expression with respect to $\phi_t$. Differentiating the quadratic expression we obtain the first order condition
\begin{align}
\label{buf}
& 2 \phi_t \hat{S}_t^2  \sigma_t^2-2\rho' \sigma \hat{S}_t \hat{Y}_t  \sigma_t \frac{\partial C}{\partial Y} \nonumber \\
&+ 2 \int_{\mathbb{R}_{+}}\left[\phi_t \hat{S}_t  (e^{\rho (1-\theta) x}-1)- (1-\theta)\left(\hat{P}(t, \sigma_{t}^2+x, V_t)-\hat{P}(t, \sigma_{t}^2, V_{t})\right)\right]\hat{S}_t  (e^{\rho (1-\theta) x}-1)\lambda  \nu_Z(dx) \nonumber \\
&+ 2 \int_{\mathbb{R}_{+}}\left[\phi_t \hat{S}_t  (e^{\rho \theta x}-1)- \theta\left(\hat{P}(t, \sigma_{t}^2+x, V_t)-\hat{P}(t, \sigma_{t}^2, V_{t})\right)\right]\hat{S}_t  (e^{\rho \theta x}-1)\lambda  \nu_Z^{(b)}(dx)=0.
\end{align}
Also, in this case the second order condition is positive, which confirms the minimization. Solution of \eqref{buf} is given by \eqref{optphit}.
\end{proof}

We conclude this section with the application of the above result to an explicit case when $P(t, \sigma_t^2, V_t)$ is given by Theorem \ref{arbitr1}.

\begin{corollary}
\label{biggy10001}
Consider the refined BN-S model given by \eqref{1}, \eqref{2new} and \eqref{4new} (with $\theta'=\theta$). Consider a European option with payoff $H(Y_T)$ where $H: \mathbb{R}_{+} \to \mathbb{R}$.
Then the risk-minimizing quadratic hedge amounts to holding a position of the underlying $S$ equal to $\phi_t= \Delta(t, S_t, Y_t)$, where
\begin{equation}
\label{optphit123}
 \Delta(t, S_t, Y_t)= \frac{\rho' \sigma \sigma_t \frac{Y_t}{S_t} \frac{\partial C}{\partial Y}+ A +B}{\sigma_t^2 +\lambda \int_{\mathbb{R}_{+}} (e^{\rho (1-\theta) x}-1)^2  \nu_Z(dx) + \lambda \int_{\mathbb{R}_{+}} (e^{\rho \theta x}-1)^2  \nu_Z^{(b)}(dx) }
\end{equation}
where $C$ is the Black-Scholes price of the option written on $Y$,and \\

$A= \frac{(1-\theta)}{S_t}e^{-r(T-t)}(1-e^{-\lambda(T-t)}) \int_{\mathbb{R}_{+}} x(e^{\rho (1-\theta) x}-1)  \nu_Z(dx) $ \\

$B= \frac{\theta}{S_t} e^{-r(T-t)}(1-e^{-\lambda(T-t)}) \int_{\mathbb{R}_{+}} x(e^{\rho \theta x}-1)  \nu_Z^{(b)}(dx)$. 
\end{corollary}
\begin{proof}
The proof follows directly with the application of Theorem \ref{biggy1} in the expressions for the $A$ and $B$ in \eqref{optphit}. 
\end{proof}

\section{Data analysis}
\label{sec3}

In this section, at first in Subsection \ref{sec31}, we present an overview of the empirical data set.  After that, in Subsection \ref{sec32}, we develop a couple of  procedures for the data analysis. Finally, the results of the data analysis and the implication of the results for the refined BN-S model are presented in Subsection \ref{sec33}. The numerical results provided in this section are primarily related to the estimation of $\theta$ value in Section \ref{sec2}. As observed in recent papers (see \cite{SWW, Wil2}) an appropriate stochastic model improves heading algorithm. Consequently, an appropriate $\theta$ value improves the refined BN-S model and thus in effect improves the quadratic hedging error as described in Subsection \ref{sub22}. The goal of this section is to develop a data-driven method to find $\theta$ from an empirical data set. 

\subsection{Description of data}
\label{sec31}

We consider crude oil price data over a period of 7 years. We use the daily Bakken crude oil price data set for the period April 4, 2012 to July 11, 2017 (Figure 1). Bakken crude oil is related to the very significant North Dakota oil boom that refers to the period of rapidly expanding oil extraction from the Bakken formation that lasted from the discovery of Parshall Oil Field in 2006, and peaked in 2012. This is the primary source for which in recent years North Dakota is always in the list of top 5 oil producing states in the United States. 

There are a total of $1,329$ available data in this set. For convenience, we index the dates (for available data)  from 0 (for April 4, 2012) to 1328 (for  July 11, 2017). The following table (Table 1) summarizes various estimates for the data set. Figures 1, 2, and 3, show various characterization of the data set.

\begin{table}[H]
\centering
\caption{Properties of the empirical data set.}
  \begin{tabular}{ | l | c | r |}
    \hline
    & Daily Price Change & Daily Price Change \% \\ \hline
    Mean & -0.03787 & -0.02183 \% \\ \hline
    Median & -0.01000 &  0.019992 \% \\ \hline
Maximum & 7.40 &  15.05 \%\\     \hline
Minimum & -7.76 &  -15.36 \%\\     \hline
  \end{tabular}
\end{table}

\begin{figure}[H]
\centering
\caption{Line plot for the Bakken oil price from April 2012- July 2017.}
\includegraphics[scale=.5]{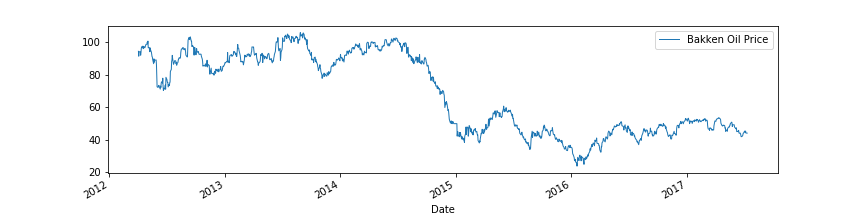}
\end{figure}

\begin{figure}[H]
\centering
\caption{Distribution plot for the Bakken oil price.}
\includegraphics[scale=.3]{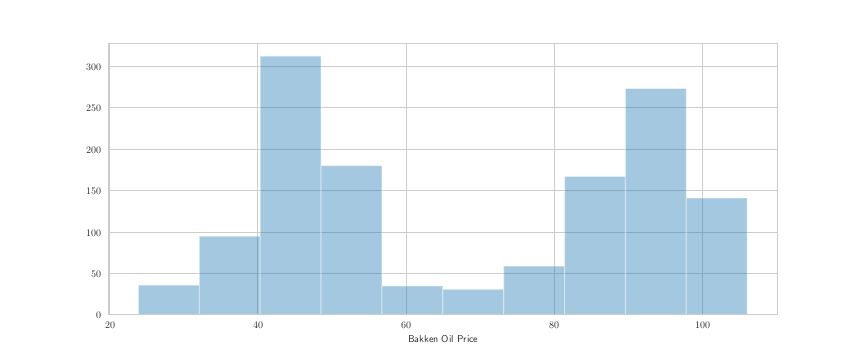}
\end{figure}

\begin{figure}[H]
\centering
\caption{Histogram for the Bakken oil price.}
\includegraphics[scale=.6]{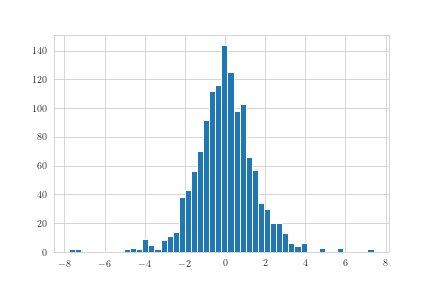}
\end{figure}

\subsection {Data analysis procedures}
\label{sec32}

For the data analysis we present here two different approaches, the aim of which is to find a $\theta$ with reasonable accuracy.
First, we implement the following procedure, naming it, \emph{Volatility Approach}, to create a \emph{classification problem} for the data set.

\begin{center}
\textbf{Volatility approach}
\end{center}

We work through the following steps (Step 1 through Step 7).
\begin{enumerate}
\item We conduct exploratory data analysis.

\item We consider the daily \textit{Bakken Oil Price} for the data, and we calculate the \textit{daily price change}  and the \textit{daily price log returns} using it. Using the \emph{daily price log returns} we calculate the \emph{realized variance} and the \emph{realized volatility}  respectively. 

\item We compute the realized volatility over 20 consecutive trading days for the oil prices. Since the computed realized volatility is very small, in order to properly utilize the volatility movements we create a new feature (column) that contains the realized volatility return in percentage, and we call it \textit{``realized volatility return in percentage"}. 

\item Using the \textit{realized volatility return in percentage} feature we  perform the following steps:
\item[4a.] We consider twenty consecutive days starting from index 0 (day 1) to index 19 (day 20). We compute the \emph{maximum} realized volatility return in percentage for those twenty trading days. We then try to identify \emph{realized volatility return in percentage} value(s), in those twenty trading days, which is strictly greater than or equal to the maximum. We assign $V=1$ if we find such values, otherwise $V=0$. 

\item[4b.] We continue step 4a for index 1 (day two) to index 20 (day twenty one) and so on respectively until we have checked through all the data points in our \emph{realized volatility return in percentage} feature. We call $V$ \emph{crash-like days}. 

\item We create a new data-frame from the old one where the features will be twenty consecutive daily change in prices. For example, if the daily change in prices are $$a_1, a_2, a_3, \cdots ,a_{18}, a_{19}, a_{20}, a_{21}, a_{22}, a_{23}, a_{24}, \cdots;$$ then the first row of the data set will contain $$a_1, a_2, a_3, \cdots ,a_{18}, a_{19}, a_{20};$$ second row of the data set will contain $$a_2, a_3, \cdots ,a_{18}, a_{19}, a_{20}, a_{21};$$

\item We create a target column for the new data-frame (as created in the preceding step) as follows: $\theta=1$ for those set of twenty Bakken oil prices that immediately precede $\mathit {at\, least\, 1 } $ (or more) crash-like days in the following twenty days. Otherwise we label the target column by $\theta=0.$

\item We run various \textit{classification algorithms} from machine learning where the input is the \textit{daily change in close price for twenty consecutive days} and output is $\theta$ -value (0 or 1). We evaluate the classification report and confusion matrix in each case.
\end{enumerate}

Figures 4-7 show various characterization of the data set related to the volatility approach described above. The purpose of the heatmap in Figure 4 is to better understand the realized volatility calculated over a period of twenty days for our entire data set. The goal is to use the numerical values and color pattern to observe any big changes for every month over the period of five years. As we can see that the realized volatility have very small values. This motivates us in computing realized volatility return in percentage. This is shown in Figure 5. For Figure 5, using the numerical values and color pattern from the heatmap we observe that over the five years the realized volatility return in percentage does not have any drastic change except for one outlier on July 2017.  Figure 6 and Figure 7 represent line plots which show us the jumps in the realized volatility return in percentage and the realized volatility over the five years, respectively. With the help of these figures we can see the highest jumps over the years, which also provides help in writing Step-4 of the above procedure.

\begin{figure}[H]
\centering
\caption{Heatmap for the realized volatility of the Bakken oil price over five years.}
\includegraphics[scale=.6]{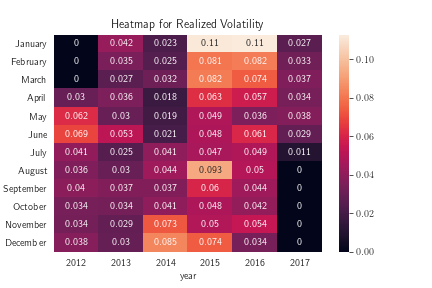}
\end{figure}

\begin{figure}[H]
\centering
\caption{Heatmap for the realized volatility return in percentage over the five years for the Bakken crude oil price. }
\includegraphics[scale=.6]{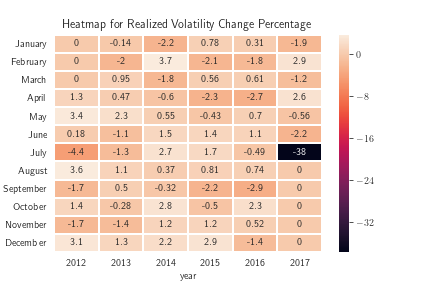}
\end{figure}

\begin{figure}[H]
\centering
\caption{Line plot for the realized volatility return in percentage for the Bakken oil price.}
\includegraphics[scale=.5]{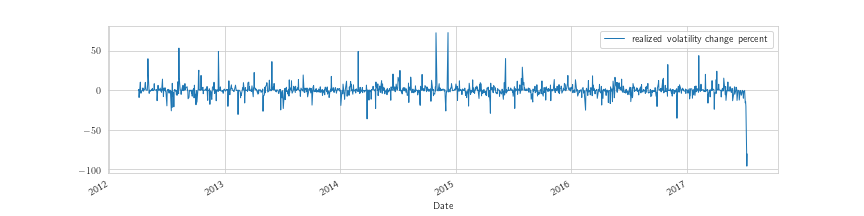}
\end{figure}

\begin{figure}[H]
\centering
\caption{Line Plot for the realized volatility of the Bakken oil price.}
\includegraphics[scale=.5]{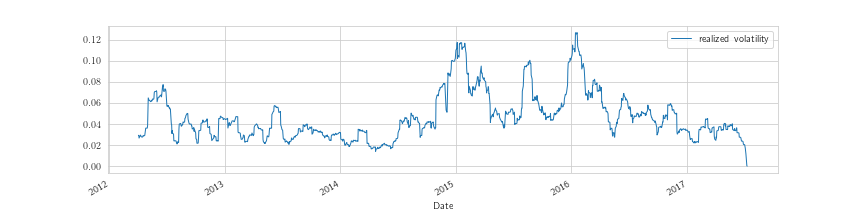}
\end{figure}

\begin{flushleft}

Next, we present the second approach to our data analysis. We implement the following procedure, naming it, \emph{Duration Approach}, to create a \emph{classification problem} for the data set.
\end{flushleft}

\begin{center}
\textbf{Duration approach}
\end{center}

We work through the following steps (Steps 1 through 7)

\begin{enumerate}
\item We conduct exploratory data analysis.

\item We consider the daily \emph{Bakken oil price} for the data. From the oil prices we calculate the \emph{daily change} and \emph{drawdowns} for the prices. A \emph{drawdown} is the total loss over consecutive days from the last maximum to the next minimum of the price. A drawdown occurring over $n$ days is described as $$d= \frac{p_{min}-p_{max}}{p_{max}}$$ with $$ p_{max} = p({t_1})> p({t_2})> \cdots > p({t_n})= p_{min}, $$ where $t_1, \cdots t_n,$ and $p(t_i)$ are the time period over \emph{n} days and \emph{Bakken oil} prices respectively.\\

It is to be noted that we will not include those prices in the drawdown calculation where the \emph{next} minimum price occurs at the beginning of the data set before the \emph{last} maximum price as well as the \emph{last} maximum price that occurs at the very end of the data set (for example- if \emph{first minimum} of the oil price is on day 5 (index 4) and \emph{first maximum} is on day 7 (index 6), we will drop the minimum price of day 5 from our computation.)

\item Our goal is to identify the dates when the drawdowns occurred in order to find the \emph{duration} of each drawdown, i.e. how long the drawdowns lasted. We will use the duration of the drawdowns as a measure to identify crash like days in our data set.

\item We fix a value for our duration, \emph{D}, and we obtain the drawdowns that lasted for that \emph{D} time period (for example, if our duration period is two days, i.e. \emph{D=2},  then we will search for drawdowns that lasted for two days (or more), and take note of their corresponding daily change prices.)

\item We create a new data-frame from the old one where the features (columns) will be ten consecutive daily change in oil prices. For example, if the daily change prices are $$a_1, a_2, a_3, \cdots , a_{8}, a_{9}, a_{10}, a_{11}, \cdots;$$ then the first row of the data set will contain $$a_1, a_2, a_3, a_4, a_5, a_6, a_7, , a_{8}, a_{9}, a_{10};$$ second row of the data set will contain $$a_2, a_3, a_{4}, a_5, a_6, a_7, a_8, a_{9}, a_{10}, a_{11};$$ etc.

\item We create a target column for the new data-frame (as created in the preceding step) as follows: $\theta=1$ for those set of ten daily change prices that immediately precede $\mathit {at\, least\, two \,drawdowns\, with\, duration\, D}$, in the following ten days. Otherwise we label the target column by $\theta=0.$

\item We run various \textit{classification algorithms} from machine learning where the input is the \textit{daily change in close price for ten consecutive days} and output is $\theta$ -value (0 or 1). We evaluate the classification report and confusion matrix in each case.
\end{enumerate}

Figures 8 and 9 show various characterization of the data set related to the duration approach described above. With the help of the bar graph in Figure 8, we can see that most drawdowns last for short period duration. For example, a more likely duration is of one or two days, compared to long duration of eight or nine days. In Figure 9, the spikes in the line plot give us some idea about the changes associated with the drawdowns in terms of duration (in number of days), over the period of five years. For example, between 2014 and 2015 most drawdowns lasted for one day or two days, and very few drawdowns went past four days.

\begin{figure}[H]
\centering
\caption{A bar graph to show the duration in the number of days for the drawdowns computed for the Bakken oil price.}
\includegraphics[scale=.4]{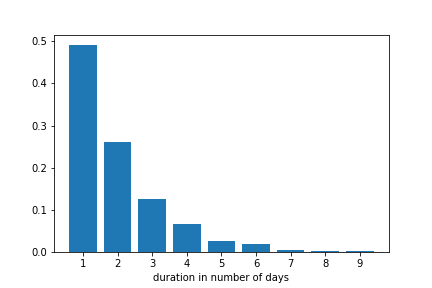}
\end{figure}

\begin{figure}[H]
\centering
\caption{A line plot to show the duration over the span of five years.}
\includegraphics[scale=.5]{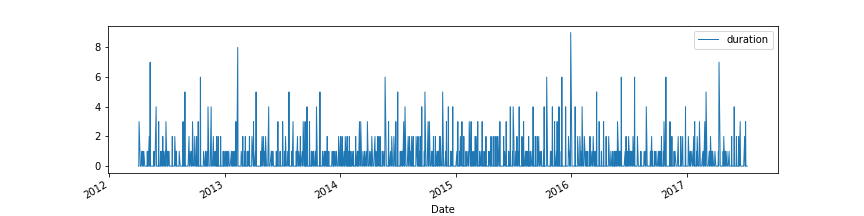}
\end{figure}

From these two approaches we will show that we can find $\theta$ with reasonable accuracy and use this for \eqref{2new}. In both the \emph{Volatility Approach} and the \emph{Duration Approach} the result can be improved by adjusting the number of days (in Step 5 for both) from twenty and ten respectively to a higher number. It is worth noting that the various deep learning models provide a value of $\theta$ between $0$ and $1$. In Step 6 (for both), we approximate that by $0$ or $1$. However, the actual value of $\theta$ may be directly used in \eqref{2new}.

We partition this data set in various ways. For each partition we use a \emph{train-test-split}, with respect to a given date.  For the analysis using the \emph{Volatility Approach} we use the \emph{maximum} to detect \emph{crash-like days} for each set of twenty data points, i.e. $\theta=1$ for the set of twenty daily change prices that immediately precede \emph{at least one crash-like days} (or more) in the following twenty days. Otherwise, we use $\theta=0$. For the analysis using the \emph{Duration Approach} we use \emph{D= 2}, i.e. $\theta=1$ for the set of ten daily change prices that immediately precede \emph{at least two drawdowns of duration D=2 } (or more) in the following ten days. Otherwise, we use $\theta=0$.

We run various \emph{supervised learning algorithms} on the crude oil price data. We begin with the \emph{logistic regression} (LR) and the \emph{random forest} (RF) classification of the data set. After that, we implement various \emph{deep learning} techniques:

\begin{enumerate}
\item[(A)] A \emph{neural network with three hidden layers} (with activation functions consisting of two $\tanh$ and one $\text{ReLU}$ respectively) and an output layer (with a $\text{softmax}$ activation function). For simplicity we approximate $\theta$ in \eqref{2new} with 0 (for example: ``duration with less than two days") and 1 (for example: ``duration with more than two days"). For these approximations, we take $\theta=1$ if the output probability for the $\text{softmax}$ activation function corresponding to $\theta=1$ is more than $0.4$.
\item[(B)] \emph{Long short-term memory} (LSTM) along with the neural network described in (A).  
\item[(C)] \emph{LSTM along with a batch normalizer} (BN) and the neural network described in (A).
\end{enumerate}

\subsection{Tables related to data analysis}
\label{sec33}

For the following tables, we provide classification reports for various machine learning algorithms. The \emph{support} is the number of samples of the true response that lie in that class. For the tables (Table 2 through Table 9), we provide classification reports for various machine learning algorithms using the Volatility Approach. For tables (Table 10 through Table 18) we provide classification reports for various machine learning algorithms using the Duration Approach. Finally, for tables (Table 19 through Table 21), we provide classification reports using both Volatility Approach and Duration Approach over the same training and testing dates.

As observed in \cite{recent}, to incorporate long range dependence, a single L\'evy subordinator is not effective for the BN-S model. If $\theta=1$ is obtained with the help of machine learning algorithms, we can modify the initial L\'evy subordinator ($Z$) with the L\'evy subordinator ($Z^{(b)}$) that corresponds to larger fluctuations. On the other hand if $\theta=0$ is obtained with the help of machine learning algorithms, we can modify the L\'evy subordinator $Z^{(b)}$ with $Z$. From the tables, it is obvious that the \emph{logistic regression} is less efficient in detecting $\theta =1$ based on the historical data. As observed in \cite{recent}, for the majority of the cases the neural network technique (A), LSTM (B), or the LSTM with a batch normalizer (C), work better than the random forest classifier. To avoid complexity, only three hidden layers are used. The results improve if the number of hidden layers is increased and also if the learning rate of the gradient descent method used is decreased.

After $\theta$ is obtained, its value can be implemented to \eqref{2new}. The machine learning algorithms can be performed on a real-time basis to continue or update with the background driving L\'evy process in the BN-S model. The analysis shows that for the Bakken oil price dynamics, the jump is \emph{not} completely stochastic. Similar to the results obtained in West Texas Intermediate (WTI or NYMEX) crude oil prices data set, as obtained in \cite{Roberts, Roberts1, recent}, there is a \emph{deterministic} component that can be implemented to apply the existing models for an extended period of time. By the deterministic component, it is meant that $\theta$ is a deterministic signal. This is deterministic in the sense that its value is extracted from the data before the model is implemented. Once the value of $\theta$ is obtained, it is kept constant for a certain period of time. Consequently, the refined BN-S model incorporates long term dependence without changing the tractability of the model.

\begin{table}[htbp]
  \centering
   \caption{Various estimations for \emph{training date(index)}: January 16, 2013(200) to June 11, 2013  (300); and \emph{testing date(index)}: June 12  (301) to July 10 (320).}
    \begin{tabular}{| c | c | c |c|c|c|}
    \hline
    & LR & RF & Neural Network (A) & LSTM (B) & BN (C) \\ \hline
Precision $\theta=0$ & 0.33  & 0.38  & 0.50   & 0.67  & 0.67 \\ \hline
Recall $\theta=0$ & 1.00     & 1.00     & 0.17  & 1.00     & 1.00 \\ \hline
f1-score $\theta=0$ & 0.50   & 0.55  & 0.25  & 0.80   & 0.80 \\ \hline
Support $\theta=0$ & 6     & 6     & 6     & 6     & 6 \\ \hline
Precision $\theta=1$ & 1.00     & 1.00     & 0.74  & 1.00     & 1.00 \\ \hline
Recall $\theta=1$ & 0.20   & 0.33  & 0.93  & 0.80   & 0.80 \\ \hline
f1-score $\theta=1$ & 0.33  & 0.50   & 0.82  & 0.89  & 0.89 \\ \hline
Support $\theta=1$ & 15    & 15    & 15    & 15    & 15 \\ \hline
    \end{tabular}%
  \label{tab:addlabel}%
\end{table}%

% Table generated by Excel2LaTeX from sheet 'VOLATILITY'
\begin{table}[htbp]
  \centering
  \caption{Various estimations for \emph{training date(index)}: February 5, 2014 (465) to June 2, 2014  (545); and \emph{testing date(index)}: June 3, 2014  (546) to July 8, 2014 (570).}
    \begin{tabular}{| c | c | c |c|c|c|}
    \hline
    & LR & RF & Neural Network (A) & LSTM (B) & BN (C) \\ \hline
    Precision $\theta=0$ & 0.32  & 0.35  & 0.17  & 0.40   & 0.36 \\ \hline
    Recall $\theta=0$ & 1.00     & 1.00     & 0.12  & 1.00     & 1.00 \\ \hline
    f1-score $\theta=0$ & 0.48  & 0.52  & 0.14  & 0.57  & 0.53 \\ \hline
    Support $\theta=0$ & 8     & 8     & 8     & 8     & 8 \\ \hline
    Precision $\theta=1$ & 1.00     & 1.00     & 0.65  & 1.00     & 1.00 \\ \hline
    Recall $\theta=1$ & 0.06  & 0.17  & 0.72  & 0.33  & 0.22 \\ \hline
    f1-score $\theta=1$ & 0.11  & 0.29  & 0.68  & 0.50   & 0.36 \\ \hline
    Support $\theta=1$ & 18    & 18    & 18    & 18    & 18 \\ \hline
    \end{tabular}%
  \label{tab:addlabel}%
\end{table}%

% Table generated by Excel2LaTeX from sheet 'VOLATILITY'
\begin{table}[htbp]
  \centering
 \caption{Various estimations for \emph{training date(index)}: June 9, 2015 (802) to December 9, 2015  (930); and \emph{testing date(index)}: December 10, 2015  (931) to February 8, 2016 (970). }
    \begin{tabular}{| c | c | c |c|c|c|}           \hline
    & LR & RF & Neural Network (A) & LSTM (B) & BN (C) \\ \hline
    Precision $\theta=0$ & 0.73  & 0.76  & 0.68  & 0.70   & 0.72 \\ \hline 
    Recall $\theta=0$ & 0.87  & 1.00     & 0.55  & 0.74  & 0.84 \\ \hline
    f1-score $\theta=0$ & 0.79  & 0.86  & 0.61  & 0.72  & 0.78 \\ \hline
    Support $\theta=0$ & 31    & 31    & 31    & 31    & 31 \\ \hline
    Precision $\theta=1$ & 0     & 0     & 0.12  & 0     & 0 \\ \hline
    Recall $\theta=1$ & 0 & 0 & 0.20   & 0 & 0 \\ \hline
    f1-score $\theta=1$ & 0     & 0     & 0.15  & 0     & 0 \\ \hline
    Support $\theta=1$ & 10    & 10    & 10    & 10    & 10 \\ \hline
    \end{tabular}%
  \label{tab:addlabel}%
\end{table}%

% Table generated by Excel2LaTeX from sheet 'VOLATILITY'
\begin{table}[htbp]
  \centering
 \caption{Various estimations for \emph{training date(index)}: November 10, 2015 (910) to March 8, 2016  (990); and \emph{testing date(index)}: March 9, 2016 (991) to April 8, 2016 (1012). }
    \begin{tabular}{| c | c | c |c|c|c|}
    \hline
    & LR & RF & Neural Network (A) & LSTM (B) & BN (C) \\ \hline
    Precision $\theta=0$ & 0.58  & 0.59  & 0.67  & 0.50   & 0.57 \\ \hline
    Recall $\theta=0$ & 0.50   & 0.93  & 0.43  & 0.64  & 0.86 \\ \hline
    f1-score $\theta=0$ & 0.54  & 0.72  & 0.52  & 0.56  & 0.69 \\ \hline
    Support $\theta=0$ & 14    & 14    & 14    & 14    & 14 \\ \hline
    Precision $\theta=1$ & 0.36  & 0     & 0.43  & 0     & 0 \\ \hline
    Recall $\theta=1$ & 0.44  & 0     & 0.67  & 0     & 0 \\ \hline
    f1-score $\theta=1$ & 0.40   & 0     & 0.52  & 0     & 0 \\ \hline
    Support $\theta=1$ & 9     & 9     & 9     & 9     & 9 \\ \hline
    \end{tabular}%
  \label{tab:addlabel}%
\end{table}%

% Table generated by Excel2LaTeX from sheet 'VOLATILITY'
\begin{table}[htbp]
  \centering
  \caption{Various estimations for \emph{training date(index)}: April 6, 2016 (1010) to August 5, 2016  (1095); and \emph{testing date(index)}: August 6, 2016  (1096) to September 6, 2016 (1116). }
    \begin{tabular}{| c | c | c |c|c|c|}
    \hline
    & LR & RF & Neural Network (A) & LSTM (B) & BN (C) \\ \hline
    Precision $\theta=0$ & 0.60   & 0.59  & 1.00     & 0.68  & 0.68 \\ \hline
    Recall $\theta=0$ & 0.92  & 1.00     & 0.08  & 1.00     & 1.00 \\ \hline
    f1-score $\theta=0$ & 0.73  & 0.74  & 0.14  & 0.81  & 0.81 \\ \hline
    Support $\theta=0$ & 13    & 13    & 13    & 13    & 13 \\ \hline
    Precision $\theta=1$ & 0.50   & 0     & 0.43  & 1.00     & 1.00 \\ \hline
    Recall $\theta=1$ & 0.11  & 0     & 1.00     & 0.33  & 0.33 \\ \hline
    f1-score $\theta=1$ & 0.18  & 0     & 0.60   & 0.50   & 0.50 \\ \hline
    Support $\theta=1$ & 9     & 9     & 9     & 9     & 9 \\ \hline
    \end{tabular}%
  \label{tab:addlabel}%
\end{table}%

% Table generated by Excel2LaTeX from sheet 'VOLATILITY'
\begin{table}[htbp]
  \centering
  \caption{Various estimations for \emph{training date(index)}:September 12, 2016 (1120) to January 12, 2017 (1205); and \emph{testing date(index)}: January 13, 2017  (1206) to February 17, 2017 (1230). }
    \begin{tabular}{| c | c | c |c|c|c|}          \hline
    & LR & RF & Neural Network (A) & LSTM (B) & BN (C) \\ \hline
    Precision $\theta=0$ & 0.58  & 0.58  & 0.50   & 0.61  & 0.53 \\ \hline
    Recall $\theta=0$ & 1.00     & 1.00     & 0.67  & 0.73  & 0.67 \\ \hline
    f1-score $\theta=0$ & 0.73  & 0.73  & 0.57  & 0.67  & 0.59 \\ \hline
    Support $\theta=0$ & 15    & 15    & 15    & 15    & 15 \\ \hline
    Precision $\theta=1$ & 0     & 0     & 0.17  & 0.50   & 0.29 \\ \hline
    Recall $\theta=1$ & 0     & 0     & 0.09  & 0.36  & 0.18 \\ \hline
    f1-score $\theta=1$ & 0     & 0     & 0.12  & 0.42  & 0.22 \\ \hline
    Support $\theta=1$ & 11    & 11    & 11    & 11    & 11 \\ \hline
    \end{tabular}%
  \label{tab:addlabel}%
\end{table}%

\begin{table}[htbp]
  \centering
  \caption{Various estimations for \emph{training date(index)}:October 24, 2016 (1150) to March 27, 2017 (1255); and \emph{testing date(index)}: March 28, 2017  (1256) to May 26, 2017 (1298) }
    \begin{tabular}{| c | c | c |c|c|c|}          \hline
    & LR & RF & Neural Network (A) & LSTM (B) & BN (C) \\ \hline
    Precision $\theta=0$ & 0.83  & 0.93  & 1.00     & 1.00     & 1.00 \\ \hline
    Recall $\theta=0$ & 0.38  & 0.69  & 0.26  & 0.13  & 0.28 \\ \hline
    f1-score $\theta=0$ & 0.53  & 0.79  & 0.41  & 0.23  & 0.44 \\ \hline
    Support $\theta=0$ & 39    & 39    & 39    & 39    & 39 \\ \hline
    Precision $\theta=1$ & 0.08  & 0.20   & 0.15  & 0.13  & 0.15 \\ \hline
    Recall $\theta=1$ & 0.40   & 0.60   & 1.00     & 1.00     & 1.00 \\ \hline
    f1-score $\theta=1$ & 0.13  & 0.30   & 0.26  & 0.23  & 0.26 \\ \hline
    Support $\theta=1$ & 5     & 5     & 5     & 5     & 5 \\ \hline
    \end{tabular}%
  \label{tab:addlabel}%
\end{table}%

% Table generated by Excel2LaTeX from sheet 'VOLATILITY'
\begin{table}[htbp]
  \centering
  \caption{Various estimations for \emph{training date(index)}: January 3, 2017 (1198) to March 20, 2017  (1250); and \emph{testing date(index)}: March 21, 2017  (1251) to April 18, 2017 (1270). }
    \begin{tabular}{| c | c | c |c|c|c|}          \hline
    & LR & RF & Neural Network (A) & LSTM (B) & BN (C) \\ \hline
    Precision $\theta=0$ & 0     & 0.50   & 0.12  & 0     & 0 \\ \hline
    Recall $\theta=0$ & 0     & 0.27  & 0.09  & 0     & 0 \\ \hline
    f1-score $\theta=0$ & 0     & 0.35  & 0.11  & 0     & 0 \\ \hline
    Support $\theta=0$ & 11    & 11    & 11    & 11    & 11 \\ \hline
    Precision $\theta=1$ & 0.15  & 0.47  & 0.23  & 0.35  & 0.27 \\ \hline
    Recall $\theta=1$ & 0.20   & 0.70   & 0.30   & 0.60   & 0.40 \\ \hline
    f1-score $\theta=1$ & 0.17  & 0.56  & 0.26  & 0.44  & 0.32 \\ \hline
    Support $\theta=1$ & 10    & 10    & 10    & 10    & 10 \\ \hline
    \end{tabular}%
  \label{tab:addlabel}%
\end{table}%

% Table generated by Excel2LaTeX from sheet 'VOLATILITY'
\begin{table}[htbp]
  \centering
  \caption{Various estimations for \emph{training date(index)}: January 16, 2013(200) to June 11, 2013  (300); and \emph{testing date(index)}: June 12  (301) to July 10 (320). }
    \begin{tabular}{| c | c | c |c|c|c|}          \hline
    & LR & RF & Neural Network (A) & LSTM (B) & BN (C) \\ \hline
Precision $\theta=0$ & 0.56  & 0.69  & 0.75  & 0.67  & 0.67 \\ \hline
Recall $\theta=0$ & 0.69  & 0.69  & 0.69  & 0.77  & 0.46 \\ \hline
f1-score $\theta=0$ & 0.62  & 0.69  & 0.72  & 0.71  & 0.55 \\ \hline
Support $\theta=0$ & 13    & 13    & 13    & 13    & 13 \\ \hline
Precision $\theta=1$ & 0.20   & 0.50   & 0.56  & 0.50   & 0.42 \\ \hline
Recall $\theta=1$ & 0.12  & 0.50   & 0.62  & 0.38  & 0.62 \\ \hline
f1-score $\theta=1$ & 0.15  & 0.50   & 0.59  & 0.43  & 0.50 \\ \hline
Support $\theta=1$ & 8     & 8     & 8     & 8     & 8 \\ \hline
    \end{tabular}%
  \label{tab:addlabel}%
\end{table}%

\begin{table}[htbp]
  \centering
  \caption{Various estimations for \emph{training date(index)}: January 14, 2014 (450) to June 16, 2014  (555); and \emph{testing date(index)}: June 17, 2014  (556) to August 12, 2014 (595). }
    \begin{tabular}{| c | c | c | c | c | c |}          \hline
& LR & RF & Neural Network (A) & LSTM (B) & BN (C) \\ \hline
Precision $\theta=0$ & 0.20   & 0.15  & 0.20   & 0.14  & 0.08 \\ \hline
Recall $\theta=0$ & 0.78  & 0.44  & 0.11  & 0.44  & 0.11 \\ \hline
f1-score $\theta=0$ & 0.32  & 0.23  & 0.14  & 0.21  & 0.10 \\ \hline
Support $\theta=0$ & 9     & 9     & 9     & 9     & 9 \\ \hline
Precision $\theta=1$ & 0.67  & 0.67  & 0.78  & 0.58  & 0.72 \\ \hline
Recall $\theta=1$ & 0.12  & 0.31  & 0.88  & 0.22  & 0.66 \\ \hline
f1-score $\theta=1$ & 0.21  & 0.43  & 0.82  & 0.32  & 0.69 \\ \hline
Support $\theta=1$ & 32    & 32    & 32    & 32    & 32 \\ \hline
    \end{tabular}%
  \label{tab:addlabel}%
\end{table}%

% Table generated by Excel2LaTeX from sheet 'DURATION'
\begin{table}[htbp]
  \centering
  \caption{Various estimations for \emph{training date(index)}: January 12, 2015 (700) to April 16, 2015 (765); and \emph{testing date(index)}: April 17, 2015  (766) to May 14, 2015 (785). }
    \begin{tabular}{| c | c | c |c|c|c|}          \hline
    & LR & RF & Neural Network (A) & LSTM (B) & BN (C) \\ \hline
    Precision $\theta=0$ & 0.27  & 0.19  & 0     & 0.17  & 0 \\ \hline
    Recall $\theta=0$ & 0.75  & 0.75  & 0     & 0.50   & 0 \\ \hline
    f1-score $\theta=0$ & 0.40   & 0.30   & 0     & 0.25  & 0 \\ \hline
    Support $\theta=0$ & 4     & 4     & 4     & 4     & 4 \\ \hline
    Precision $\theta=1$ & 0.90   & 0.80   & 0.79  & 0.78  & 0.78 \\ \hline
    Recall $\theta=1$ & 0.53  & 0.24  & 0.88  & 0.41  & 0.82 \\ \hline
    f1-score $\theta=1$ & 0.67  & 0.36  & 0.83  & 0.54  & 0.80 \\ \hline
    Support $\theta=1$ & 17    & 17    & 17    & 17    & 17 \\ \hline
    \end{tabular}%
  \label{tab:addlabel}%
\end{table}%

% Table generated by Excel2LaTeX from sheet 'DURATION'
\begin{table}[htbp]
  \centering
  \caption{Various estimations for \emph{training date(index)}: May 14, 2015 (785) to November 17, 2015  (915); and \emph{testing date(index)}: November 18, 2015  (916) to December 16, 2015 (935). }
    \begin{tabular}{| c | c | c |c|c|c|}          \hline
    & LR & RF & Neural Network (A) & LSTM (B) & BN (C) \\ \hline
    Precision $\theta=0$ & 0.43  & 0.54  & 0.40   & 0.38  & 0.42 \\ \hline
    Recall $\theta=0$ & 0.67  & 0.78  & 0.22  & 0.56  & 0.56 \\ \hline
    f1-score $\theta=0$ & 0.52  & 0.64  & 0.29  & 0.45  & 0.48 \\ \hline
    Support $\theta=0$ & 9     & 9     & 9     & 9     & 9 \\ \hline
    Precision $\theta=1$ & 0.57  & 0.75  & 0.56  & 0.50   & 0.56 \\ \hline
    Recall $\theta=1$ & 0.33  & 0.50   & 0.75  & 0.33  & 0.42 \\ \hline
    f1-score $\theta=1$ & 0.42  & 0.60   & 0.64  & 0.40   & 0.48 \\ \hline
    Support $\theta=1$ & 12    & 12    & 12    & 12    & 12 \\ \hline
    \end{tabular}%
  \label{tab:addlabel}%
\end{table}%

% Table generated by Excel2LaTeX from sheet 'DURATION'
\begin{table}[htbp]
  \centering
  \caption{Various estimations for \emph{training date(index)}: January 8, 2016 (950) to March 8, 2016  (990); and \emph{testing date(index)}: March 9, 2016  (991) to March 30, 2016 (1005).}
    \begin{tabular}{| c | c | c |c|c|c|}          \hline
    & LR & RF & Neural Network (A) & LSTM (B) & BN (C) \\ \hline
    Precision $\theta=0$ & 0.75  & 0.67  & 0.80   & 0.89  & 1.00 \\ \hline
    Recall $\theta=0$ & 0.82  & 0.73  & 0.36  & 0.73  & 0.73 \\ \hline
    f1-score $\theta=0$ & 0.78  & 0.70   & 0.50   & 0.80   & 0.84 \\ \hline
    Support $\theta=0$ & 11    & 11    & 11    & 11    & 11 \\ \hline
    Precision $\theta=1$ & 0.50   & 0.25  & 0.36  & 0.57  & 0.62 \\ \hline
    Recall $\theta=1$ & 0.40   & 0.20   & 0.80   & 0.80   & 1.00 \\ \hline
    f1-score $\theta=1$ & 0.44  & 0.22  & 0.50   & 0.67  & 0.77 \\ \hline
    Support $\theta=1$ & 5     & 5     & 5     & 5     & 5 \\ \hline
    \end{tabular}%
  \label{tab:addlabel}%
\end{table}%

% Table generated by Excel2LaTeX from sheet 'DURATION'
\begin{table}[htbp]
  \centering
 \caption{Various estimations for \emph{training date(index)}: April 1, 2016 (1007) to August 5, 2016  (1095); and \emph{testing date(index)}: August 6, 2016  (1096) to October 3, 2016 (1135). }
    \begin{tabular}{| c | c | c |c|c|c|}          \hline
    & LR & RF & Neural Network (A) & LSTM (B) & BN (C) \\ \hline
    Precision $\theta=0$ & 0.85  & 0.85  & 0.91  & 0.83  & 0.84 \\ \hline
    Recall $\theta=0$ & 0.97  & 0.97  & 0.60   & 0.83  & 0.91 \\ \hline
    f1-score $\theta=0$ & 0.91  & 0.91  & 0.72  & 0.83  & 0.88 \\ \hline
    Support $\theta=0$ & 35    & 35    & 35    & 35    & 35 \\ \hline
    Precision $\theta=1$ & 0     & 0     & 0.22  & 0     & 0 \\ \hline
    Recall $\theta=1$ & 0     & 0     & 0.67  & 0     & 0 \\ \hline
    f1-score $\theta=1$ & 0     & 0     & 0.33  & 0     & 0 \\ \hline
    Support $\theta=1$ & 6     & 6     & 6     & 6     & 6 \\ \hline
    \end{tabular}%
  \label{tab:addlabel}%
\end{table}%

% Table generated by Excel2LaTeX from sheet 'DURATION'
\begin{table}[htbp]
  \centering
  \caption{Various estimations for \emph{training date(index)}: September 12, 2016 (1120) to December 13, 2016  (1185); and \emph{testing date(index)}:  December 14, 2016 (1185) to January 20, 2017 (1210). }
    \begin{tabular}{| c | c | c |c|c|c|}          \hline
    & LR & RF & Neural Network (A) & LSTM (B) & BN (C) \\ \hline
    Precision $\theta=0$ & 0.33  & 0.38  & 0.20   & 0.21  & 0.18 \\ \hline
    Recall $\theta=0$ & 0.40   & 0.50   & 0.20   & 0.30   & 0.30 \\ \hline
    f1-score $\theta=0$ & 0.36  & 0.43  & 0.20   & 0.25  & 0.22 \\ \hline
    Support $\theta=0$ & 10    & 10    & 10    & 10    & 10 \\ \hline
    Precision $\theta=1$ & 0.57  & 0.62  & 0.50   & 0.42  & 0.22 \\ \hline
    Recall $\theta=1$ & 0.50   & 0.50   & 0.50   & 0.31  & 0.12 \\ \hline
    f1-score $\theta=1$ & 0.53  & 0.55  & 0.50   & 0.36  & 0.16 \\ \hline
    Support $\theta=1$ & 16    & 16    & 16    & 16    & 16 \\ \hline
    \end{tabular}%
  \label{tab:addlabel}%
\end{table}%

% Table generated by Excel2LaTeX from sheet 'DURATION'
\begin{table}[htbp]
  \centering
  \caption{Various estimations for \emph{training date(index)}: November 1, 2016 (1156) to February 1, 2017 (1218); and \emph{testing date(index)}: February 2, 2017  (1218) to March 1, 2017  (1237). }
    \begin{tabular}{| c | c | c |c|c|c|}          \hline
    & LR & RF & Neural Network (A) & LSTM (B) & BN (C) \\ \hline
    Precision $\theta=0$ & 0.80   & 0.76  & 0.86  & 0.93  & 1.00 \\ \hline
    Recall $\theta=0$ & 0.86  & 0.93  & 0.43  & 0.93  & 0.14 \\ \hline
    f1-score $\theta=0$ & 0.83  & 0.84  & 0.57  & 0.93  & 0.25 \\ \hline
    Support $\theta=0$ & 14    & 14    & 14    & 14    & 14 \\ \hline
    Precision $\theta=1$ & 0.60   & 0.67  & 0.38  & 0.83  & 0.33 \\ \hline
    Recall $\theta=1$ & 0.50   & 0.33  & 0.83  & 0.83  & 1.00 \\ \hline
    f1-score $\theta=1$ & 0.55  & 0.44  & 0.53  & 0.83  & 0.50 \\ \hline
    Support $\theta=1$ & 6     & 6     & 6     & 6     & 6 \\ \hline
    \end{tabular}%
  \label{tab:addlabel}%
\end{table}%

% Table generated by Excel2LaTeX from sheet 'DURATION'
\begin{table}[htbp]
  \centering
 \caption{Various estimations for \emph{training date(index)}: January 5, 2017 (1200) to May 31, 2017  (1300); and \emph{testing date(index)}: June 1, 2017 (1301) to July 11, 2017 (1328). }
    \begin{tabular}{| c | c | c |c|c|c|}           \hline
    & LR & RF & Neural Network (A) & LSTM (B) & BN (C) \\ \hline
    Precision $\theta=0$ & 0.75  & 0.73  & 0.75  & 0.71  & 0.57 \\ \hline
    Recall $\theta=0$ & 1.00     & 0.73  & 0.40   & 0.67  & 0.27 \\ \hline
    f1-score $\theta=0$ & 0.86  & 0.73  & 0.52  & 0.69  & 0.36 \\ \hline
    Support $\theta=0$ & 15    & 15    & 15    & 15    & 15 \\ \hline
    Precision $\theta=1$ & 0     & 0.20   & 0.25  & 0.17  & 0.15 \\ \hline
    Recall $\theta=1$ & 0     & 0.20   & 0.60   & 0.20   & 0.40 \\ \hline
    f1-score $\theta=1$ & 0     & 0.20   & 0.35  & 0.18  & 0.22 \\\hline
    Support $\theta=1$ & 5     & 5     & 5     & 5     & 5 \\ \hline
    \end{tabular}%
  \label{tab:addlabel}%
\end{table}%

\begin{table}[htbp]
\centering
\caption{Various estimations for \emph{training date(index)}: April 18, 2012 (10) to August 16, 2012  (95); and \emph{testing date(index)}: August 17, 2012 (96) to September 28, 2012(125)}

    \begin{tabular}{ | c | c | c | c | c | c |}          \hline
    \multicolumn{6}{|c|}{Volatility Approach}\\ \hline
     & LR & RF & Neural Network (A) & LSTM (B) & BN (C) \\ \hline
Precision $\theta=0$ & 0.39  & 0.40   & 0.42  & 0.44  & 0.39 \\ \hline
Recall $\theta=0$ & 1.00     & 1.00     & 0.83  & 1.00     & 1.00 \\ \hline
f1-score $\theta=0$ & 0.56  & 0.57  & 0.56  & 0.62  & 0.56 \\ \hline
Support $\theta=0$ & 12    & 12    & 12    & 12    & 12 \\ \hline
Precision $\theta=1$ & 0     & 1.00     & 0.71  & 1.00     & 0 \\ \hline
Recall $\theta=1$ & 0     & 0.05  & 0.26  & 0.21  & 0 \\ \hline
f1-score $\theta=1$ & 0     & 0.10   & 0.38  & 0.35  & 0 \\ \hline
Support $\theta=1$ & 19    & 19    & 19    & 19    & 19 \\ \hline
                 %&       &       &       &       &  \\ \hline
            
\multicolumn{6}{| c |}{Duration Approach} \\         \hline
 & LR & RF & Neural Network (A) & LSTM (B) & BN (C) \\ \hline
Precision $\theta=0$ & 0.39  & 0.41  & 0.40   & 0.38  & 0.39 \\ \hline
Recall $\theta=0$ & 0.85  & 0.92  & 0.46  & 0.77  & 0.85 \\ \hline        f1-score $\theta=0$ & 0.54  & 0.57  & 0.43  & 0.51  & 0.54 \\ \hline
Support $\theta=0$ & 13    & 13    & 13    & 13    & 13 \\ \hline
Precision $\theta=1$ & 0.33  & 0.50   & 0.56  & 0.40   & 0.33 \\ \hline
Recall $\theta=1$ & 0.06  & 0.06  & 0.50   & 0.11  & 0.06 \\ \hline
f1-score $\theta=1$ & 0.10  & 0.10   & 0.53  & 0.17  & 0.10 \\ \hline
Support $\theta=1$ & 18    & 18    & 18    & 18    & 18 \\ \hline

    \end{tabular}%
  \label{tab:addlabel}%
\end{table}%

\begin{table}[htbp]
  \centering
  \caption{Various estimations for \emph{training date(index)}: August 23, 2012 (100) to February 14, 2013  (220); and \emph{testing date(index)}: February 15, 2013  (221) to April 15, 2013 (260). }
  
\begin{tabular}{| c| c | c | c |c|c|c|}          \hline
 \multicolumn{6}{|c|}{Volatility Approach}\\ \hline
 & LR & RF & Neural Network (A) & LSTM (B) & BN (C) \\ \hline
Precision $\theta=0$ & 0.76  & 0.89  & 0.82  & 0.96  & 0.97 \\ \hline
Recall $\theta=0$ & 0.91  & 0.97  & 0.28  & 0.72  & 0.94 \\ \hline
f1-score $\theta=0$ & 0.83  & 0.93  & 0.42  & 0.82  & 0.95 \\ \hline
Support $\theta=0$ & 32    & 32    & 32    & 32    & 32 \\ \hline
Precision $\theta=1$ & 0.00     & 0.83  & 0.23  & 0.47  & 0.80 \\ \hline
Recall $\theta=1$ & 0.00     & 0.56  & 0.78  & 0.89  & 0.89 \\ \hline
f1-score $\theta=1$ & 0.00     & 0.67  & 0.36  & 0.62  & 0.84 \\ \hline
Support $\theta=1$ & 9     & 9     & 9     & 9     & 9 \\ \hline
        %  &       &       &       &       &       &  \\ \hline
\multicolumn{6}{| c |}{Duration Approach} \\         \hline
Precision $\theta=0$ & 0.45  & 0.50   & 0.92  & 0.48  & 0.57 \\ \hline
Recall $\theta=0$ & 0.68  & 0.77  & 0.50   & 0.45  & 0.59 \\ \hline
f1-score $\theta=0$ & 0.55  & 0.61  & 0.65  & 0.47  & 0.58 \\ \hline
Support $\theta=0$ & 22    & 22    & 22    & 22    & 22 \\ \hline
Precision $\theta=1$ & 0.12  & 0.29  & 0.62  & 0.40   & 0.50 \\ \hline
Recall $\theta=1$ & 0.05  & 0.11  & 0.95  & 0.42  & 0.47 \\ \hline
f1-score $\theta=1$ & 0.07  & 0.15  & 0.75  & 0.41  & 0.49 \\ \hline
Support $\theta=1$ & 19    & 19    & 19    & 19    & 19 \\ \hline
    \end{tabular}%
  \label{tab:addlabel}%
\end{table}%

\begin{table}[htbp]
  \centering
  \caption{Various estimations for \emph{training date(index)}: August 21,2013 (350) to March 27,2014  (500); and \emph{testing date(index)}: March 28,2014  (501) to May 23,2014 (540).}
    \begin{tabular}{|c | c | c | c |c|c|c|}           \hline
     \multicolumn{6}{|c|}{Volatility Approach}\\ \hline
& LR & RF & Neural Network (A) & LSTM (B) & BN (C) \\ \hline
Precision $\theta=0$ & 0.73  & 0.75  & 0.77  & 0.86  & 0.84 \\ \hline
Recall $\theta=0$ & 1.00     & 1.00     & 0.67  & 0.83  & 0.87 \\ \hline
f1-score $\theta=0$ & 0.85  & 0.86  & 0.71  & 0.85  & 0.85 \\ \hline
Support $\theta=0$ & 30    & 30    & 30    & 30    & 30 \\ \hline
Precision $\theta=1$ & 0     & 1.00     & 0.33  & 0.58  & 0.60 \\ \hline
Recall $\theta=1$ & 0     & 0.09  & 0.45  & 0.64  & 0.55 \\ \hline
f1-score $\theta=1$ & 0     & 0.17  & 0.38  & 0.61  & 0.57 \\ \hline
Support $\theta=1$ & 11    & 11    & 11    & 11    & 11 \\ \hline
          %&       &       &       &       &       &  \\ \hline
        \multicolumn{6}{| c |}{Duration Approach} \\         \hline
Precision $\theta=0$ & 0.46  & 0.49  & 0.35  & 0.48  & 0.58 \\ \hline
Recall $\theta=0$ & 0.95  & 1.00     & 0.32  & 0.79  & 0.74 \\ \hline
f1-score $\theta=0$ & 0.62  & 0.66  & 0.33  & 0.60   & 0.65 \\ \hline
Support $\theta=0$ & 19    & 19    & 19    & 19    & 19 \\ \hline
Precision $\theta=1$ & 0.50   & 1.00     & 0.46  & 0.60   & 0.71 \\ \hline
Recall $\theta=1$ & 0.05  & 0.09  & 0.50   & 0.27  & 0.55 \\ \hline
f1-score $\theta=1$ & 0.08  & 0.17  & 0.48  & 0.37  & 0.62 \\ \hline
Support $\theta=1$ & 22    & 22    & 22    & 22    & 22 \\ \hline
    \end{tabular}%
  \label{tab:addlabel}%
\end{table}%


\pagebreak
\section{Conclusion}
\label{sec4}

Management of oil revenue is risky in recent years as the volatility of oil prices has increased significantly in the last several years. Firms and organizations deal with these risks in different ways.  A refined version of the major tractable stochastic model- the BN-S model- is implemented in the present paper for minimizing the quadratic hedging error. As shown in this paper, there are certain advantages of this model relative to traditional and conventional appropriates. The theoretical results are implemented for the data analysis of the Bakken oil price. But, the procedure and analysis presented in this paper, in principle, can also be performed to other financial commodities. The procedure presented in this paper also shows a data science driven approach to deal with the stochastic models for the commodity market. It is shown that a data science driven approach can be used to effectively modify stochastic models. The resulting model can be enacted to better analyze the commodity markets. 

The two approaches discussed in the data analysis section of this paper are attempts to identify crash-like days. At the same time, it portrays the potential of merging the data science with stochastic models. In this paper, we apply various supervised and deep learning techniques to identify $\theta$ by working in conjunction with realized volatility and duration of drawdown of oil prices, respectively. Nonetheless, there is still room for further refinement of these discussed approaches. For the \emph{Volatility Approach}, rather than simply looking at the \emph{maximum realized volatility return in percentage} over a period of twenty consecutive days in Step-4 of this particular approach, one can look into the \emph{mean of the realized volatility return in percentage} over twenty consecutive days, or look at the highest positive jump in \emph{realized volatility} over five years. These will result in two different approaches. As it can be observed from this paper, data science driven approaches, and especially deep learning techniques can be a valuable resource into effective modification and efficient analysis of the stochastic models for the commodity market.   
 \\

\textbf{DATA AVAILABILITY STATEMENT}:  The data that support the findings of this study are available on request from the corresponding author. The data are not publicly available due to privacy or ethical restrictions.  \\

\textbf{Acknowledgment}: The authors would like to thank the anonymous reviewers for their careful reading of the manuscript and for suggesting points to improve the quality of the paper.

%\textbf{Declarations}:

%\begin{itemize}
%\item \textbf{Data availability statement}: The data that support the findings of this study are available in the public domain: \\ \url{https://www.macrotrends.net/1369/crude-oil-price-history-chart}. 
%\item \textbf{Competing interests}: The authors declare that they have no competing interests. 

%\item \textbf{Funding}: None. 
%\item \textbf{Authors' contributions}: All authors contributed to the idea and concept of the paper. All authors read and approved the final manuscript. 
%\item \textbf{Acknowledgements}: None. 
%\end{itemize}

\end{document}